\documentclass[runningheads]{llncs}
\usepackage{graphicx}
\usepackage{epsfig}
\usepackage{graphics}
\usepackage{algorithm}
\usepackage{algorithmic}
\usepackage{tikz}
\usepackage{epstopdf}
\usepackage{fullpage}
\usepackage{enumerate}
\usepackage{lineno,comment}
\usepackage{amsmath,amssymb}
\date{}

\newtheorem{obs}{Observation}
\title{\emph{2-}Trees: Structural Insights and the study of Hamiltonian Paths}
\author{ P.Renjith and N.Sadagopan } 
\institute{Department of Computer Engineering,\\ Indian Institute of Information Technology, Design and Manufacturing, Kancheepuram, Chennai, India. \\
\email{$\{coe14d002,sadagopan\}@iiitdm.ac.in$}}
\begin{document}
\maketitle
\begin{abstract}
\noindent
For a connected graph, a path containing all vertices is known as \emph{Hamiltonian path}.  For general graphs, there is no known necessary and sufficient condition for the existence of Hamiltonian paths and the complexity of finding a Hamiltonian path in general graphs is NP-Complete.  We present a necessary and sufficient condition for the existence of Hamiltonian paths in 2-trees.  Using our characterization, we also present a linear-time algorithm for the existence of Hamiltonian paths in 2-trees.  Our characterization is based on a deep understanding of the structure of 2-trees and the combinatorics presented here may be used in other combinatorial problems restricted to 2-trees. 
\\
\end{abstract}
\section{Introduction}
Hamiltonian path (cycle) problem is one of the most extensively studied problem, that looks for a spanning path (cycle) in a connected graph.  Interestingly, such a problem has many applications in real life, related to  medical genetic studies\cite{app1}, for chromosome studies, in physics \cite{app2} and operational research \cite{app3}. Hamiltonian problem is one among the NP-complete problems in general graphs\cite{karp}.  
For a graph, the fundamental research question is to find a necessary and sufficient condition for the existence of a Hamiltonian path (Hamiltonian cycle).  Surprisingly, there is no known necessary and sufficient condition despite many attempts from several researchers \cite{dirac,newman,ore,nash,chvatal,Gould,ainouche,Faudree,chvatalerdos,bondychvatal}.  However, there are well-known {\em necessary conditions } and {\em sufficient conditions}.  
Necessary condition by V.Chvatal \cite{west} states that if a connected graph $G$ has a Hamiltonian cycle, then for each non-empty subset $S\subset V(G)$, the graph $G-S$ has at most  $|S|$ components.
Sufficient condition looks for structural conditions for a graph to have a Hamiltonian cycle, mostly
for the presence of higher degree vertices in a graph.  Sufficient conditions based on vertex degree has been proposed in the literature \cite{dirac,newman,ore,nash,chvatal,Gould,ainouche,Faudree}. Other sufficient conditions based on graph closure, independence number and connectivity have also been formulated \cite{chvatalerdos,bondychvatal}.\\\\
Interestingly, several {\em variants} of Hamiltonian path (Hamiltonian cycle) have been looked at in the past by imposing appropriate constraints.  A graph is said to be homogeneously traceable, if there exist a Hamiltonian path beginning at every vertex of $G$.  A hypo-Hamiltonian graph is a non-Hamiltonian graph $G$ such that $G-v$ is Hamiltonian for every vertex, $ v \in V(G)$.  Existence of homogeneously traceable graph and hypo-Hamiltonian graph \cite{herz,lindergen,chatrand} were studied in the literature.  A graph is $k-$ordered Hamiltonian if for every ordered sequence of $k$ vertices, there exists a Hamiltonian cycle that encounters the vertices of the sequence in the given order.  If there exist a Hamiltonian path between every pair of vertices then the graph is called Hamiltonian connected.  A pancyclic graph on $n$ vertices is a graph which has every cycle of length $l, 3 \leq l \leq n$.   Sufficient conditions for the existence of  $k-$ordered Hamiltonian, Hamiltonian connected, and pancyclic graphs, similar to Ore's and Dirac's results have also been proposed in the literature \cite{schultz,keirstead,faudree1,faudree2}.\\\\
On the algorithmic front, it is well-known that Hamiltonian path (Hamiltonian cycle) is NP-complete.  When a combinatorial problem is NP-complete in general graphs, it is natural to study the complexity on restricted graph classes or special graph classes.  The popular graph classes studied in the literature are chordal, interval, grid, chordal bipartite, distance hereditary,  circular arc, cubic, and planar.  It is proved  that Hamiltonian problem is NP-complete on  various restricted graph classes like chordal \cite{bertossi}, grid \cite{Gordon},  chordal-bipartite \cite{muller}, planar \cite{tarjanplanar}, bipartite \cite{akiyama}, directed path graph \cite{giri} and rooted directed path graph \cite{pandapradhan}.  On the other side, nice polynomial-time algorithms for the same has been found on interval \cite{keil,hungint}, circular arc \cite{hungcir,shih}, proper interval \cite{panda,ibarra}, distance hereditary\cite{hungdis}, and specific sub class of grid graphs \cite{collins}.  Nice structural characterization for the existence of Hamiltonian cycle in claw $(K_{1,3})$-free graphs \cite{goodman,gould1,gould3,boersma,bedrossian,gould2,Ryjacek} has been studied in the past as well.  A detailed survey on the Hamiltonian properties has been compiled by Broersma and Gould \cite{s1,s2,s3}.  \\\\
  \emph{Chordal graphs} are one among the restricted graph classes possessing nice structural characteristics.  A graph is said to be chordal if every cycle of length more than three has a chord.  A chord is an edge joining  two non-consecutive vertices of a cycle.  Given that chordal graphs have polynomial-time algorithm on various classical combinatorial problems such as vertex cover, clique, it is natural to investigate the complexity of Hamiltonian problems on chordal graphs.  As already mentioned, Hamiltonian cycle problem on chordal graphs is NP-complete, this brings our focus on some subclasses of chordal graphs.  Interestingly, interval graphs, a quite popular subclass of chordal graphs have a polynomial-time algorithm for Hamiltonian problem.  Similarly, other special graph classes like proper-interval graphs and circular arc graphs also possess polynomial-time algorithms.  To the best of our knowledge, these are the only polynomial-time results for Hamiltonian cycle problem on the sub class of chordal graphs. \\\\
 The objective of this paper is two fold.  First, we present structural insights on 2-trees.  Further, we present a necessary and sufficient condition for the existence of Hamiltonian paths, and using the characterization, a polynomial-time algorithm to obtain Hamiltonian paths in 2-trees is also presented.\\ 
\\ \textbf{Our Approach:} Given a 2-tree $G$, we perform a series of computations to obtain a Hamiltonian path.  We first check whether $G$ is $3$-pyramid free.  If so, we output a Hamiltonian path.  We next check whether $G$ is $4$-pyramid free and contains exactly one $3$-pyramid.  If so, $G$ contains a Hamiltonian path.  If $G$ is $4$-pyramid free and contains at least two $3$-pyramids, then we first perform a pruning of the 2-tree by removing 2-degree vertices iteratively satisfying some structural condition.  During pruning, we also color the edges, in particular if an edge $e$ in $G$ is colored blue during pruning, it indicates that there is a $3$-pyramid free sub 2-tree with $e$ as the base 2-tree.  We also observe that the first level pruning yields a $3$-pyramid free 2-tree with some edges are colored blue.  On this pruned 2-tree we identify five sets of edges (non-blue edges) which will be removed from $G$.  The existence of Hamiltonian path in $G$ is determined based on some structural conditions on this simplified graph.  We also highlight that each pruning step is a solution preserving step and indeed guarantees a Hamiltonian path.
\\ \textbf{Road Map:} We next present graph preliminaries.  In Section 2, we present a necessary and sufficient condition for a 2-tree to have Hamiltonian paths and Hamiltonian cycles.  The algorithm for finding a Hamiltonian path in a 2-tree is presented in Section 2.4.  
\subsection{Graph Preliminaries}
Notation is as per \cite{bondy}. In this paper we work with simple, connected, unweighted graphs.  
For a graph $G$ the vertex set is $V(G)$ and the edge set is $E(G)=\{uv:u,v \in V(G)$ and $u$ is adjacent to $v$ in $G$ and $u\neq v\}$.  
The neighborhood of vertex $v$ is $N_{G}(v)  = \{u:uv\in E(G)\}$.  
The degree of a vertex $v$ is $d_{G}(v) = |N_{G}(v)|$.   
$\Delta(G)$ denotes the maximum degree in $G$.  
For a vertex $u$, \emph{close(u)}$ = \{vw:uv, uw\in E(G)\}$.  
For an edge $e=uv$, \emph{close(e) }$=\{w:wu, wv\in E(G)\}$.
A \emph{2-}tree G can be inductively constructed as follows.  
An edge is a \emph{2-}tree.  If G is a \emph{2-}tree on $(n-1), n\geq 3$ vertices, then select an edge $uv\in E(G)$ and add a vertex $z$ to $G$ such that $N_{G^{'}}(z) = \{u,v\}$; $uz, vz\in E(G^{'})$ is also a \emph{2-}tree $G^{'}$ on $n$ vertices.
We call a  \emph{2-}tree $n$-$pyramid$, $n\geq 2$ if it has $n+2$ vertices and an edge $\{u,v\}$ such that $|N_{G}(u) \cap N_{G}(v)| = n$.  A $5$-pyramid is shown in Figure \ref{fig1}.  We call a 2-tree $G$, $n$-pyramid free if $G$ contains no $n$-pyramid as an induced subgraph.\\
$K_{n}$ denotes a complete graph on $n$ vertices.  
A vertex $v\in V(G)$ is called a \emph{simplicial vertex} if $N_{G}(v)$ induces a complete subgraph of $G$ \cite{golumbic}.  
\emph{Perfect vertex Elimination Ordering} (\emph{PEO}) is an ordering of the vertices of a graph as ($v_{1},v_{2},\ldots,v_{n}$) such that each $v_{i}$ is a simplicial vertex of the induced subgraph on vertices $\{v_{i},v_{i+1},\ldots,v_{n}\}$.  Note that by definition \emph{2-}trees are chordal. 
\begin{figure}[h!]
\begin{center}
	\includegraphics[scale=0.25]{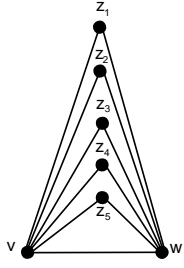} 
\caption{$5$-pyramid $2$-tree}
\label{fig1}
\end{center}
\end{figure}
An \emph{(s,t)-Hamiltonian path} is a Hamiltonian path from $s$ to $t$.  If $S\subset V(G)$, then the induced subgraph $G[V(G)\backslash S]$ is represented as $G-S$.  For $G-\{v\}$, we also use $G-v$.  $G$ is $H$-free if $G$ does not contain $H$ as an induced subgraph.  An \emph{edge-induced} subgraph $H$ of $G$ is formed on the edge set $E(H)\subset E(G)$ and $V(H)=\{u:$edge $uv\in E(H)$ is incident on the vertex $u\}$. 
$c(G)$ denotes the number of connected components in the graph $G$.  
For a connected graph $G$, $c(G)$=$1$.  $S\subset V(G)$ is a \emph{vertex separator} if $c(G)<c(G-S)$.  
A cut vertex $v \in V(G)$ is a vertex such that $c(G)<c(G-v)$.  $C_v$ represents the component of a disconnected graph containing a vertex $v$.
A $2$-connected component is a component without a cut vertex.  
A \emph{block} is a maximal 2-connected component of a graph.  
Path from vertex $u$ to $v$, $P_{uv} $ is represented as $(u, u_{1}, u_{2}, \ldots, u_{k}, v), k\geq 0$, where vertices $u_{i}, 1\leq i\leq k$ are termed as internal vertices of $P_{uv}$.  We use $P_{uv}$ to represent $V(P_{uv})$ and hence $|V(P_{uv})|=|P_{uv}|$.   \emph{Blue path} is a path with all its edges blue.  
\section{Structural insights into 2-trees}
In this section we shall present some insights into the structure of 2-trees.
Below observation is a well-known characteristics of any 2-tree.
 \begin{obs} \label{obs1}
 Let $G$ be a 2-tree.  $G$ forbids $K_{n\ge4}$, and $C_{n\ge4}$ as an induced subgraph.
 \end{obs}
\begin{lemma}\label{lem1}
Let $G$ be a 2-tree and $uv\in E(G)$.  If $|N_{G}(u)\cap N_{G}(v)| = n$, then $c(G-\{u,v\}) = n$.
\end{lemma}
\begin{proof}
We use induction on $n=|V(G)|$.  The claim is immediate for $n\le 2$.  For $n\ge3$, let $w$ be a simplicial vertex in $G_n$
such that $N_{G_n}(w)=\{x,y\}$.  From the induction hypothesis, in $G_{n-1}=G_n-\{w\}$, for every $uv\in E(G_{n-1})$, $c(G_{n-1}-\{u,v\}) = |N_{G_{n-1}}(u)\cap N_{G_{n-1}}(v)|$.  Clearly, for every $uv\in E(G_n)\setminus\{xy\}$, $c(G_{n}-\{u,v\})=c(G_{n-1}-\{u,v\})=|N_{G_{n}}(u)\cap N_{G_{n}}(v)|$, and $c(G_{n}-\{x,y\})=c(G_{n-1}-\{x,y\})+1=|N_{G_{n-1}}(x)\cap N_{G_{n-1}}(y)|+1=|N_{G_{n}}(x)\cap N_{G_{n}}(y)|$.  This completes the induction. $\hfill \qed$ 
\end{proof} 
\begin{theorem}[Chvatal \cite{west}] \label{thmchvatal}
If a graph $G$ has a Hamiltonian cycle, then for every $S \subset V(G)$, $c(G-S) \leq |S|$.
\end{theorem}
Theorem \ref{thmchvatal} is a well-known necessary condition for Hamiltonicity in general graphs.  
Also there is no necessary and sufficient condition for hamiltonicity in general graphs.
In Theorem \ref{thm1}, we present a necessary and sufficient condition for the existence of Hamiltonian cycles in 2-trees.  
Further, we show that Theorem \ref{thmchvatal} is indeed sufficient for 2-trees, which we establish using Theorem $\ref{thm1}$, and Theorem $\ref{thm2}$.  
\begin{obs}
For every $k\geq 4$, any $3$-pyramid free 2-tree is also a $k$-pyramid free 2-tree. 
\end{obs}
\begin{theorem}\label{thm1}
Let G be a 2-tree.  G has a Hamiltonian cycle if and only if $G$ is $3$-pyramid free.
\end{theorem}
\begin{proof} \emph{Necessity:}
Assume for a contradiction that $G$ has a $3$-pyramid.  This implies there exist $uv \in E(G)$ such that $|N_{G}(u) \cap N_{G}(v)| = 3$.  By Lemma \ref{lem1}, $c(G-\{u,v\}) =3$.  Further, by Theorem \ref{thmchvatal}, $G$ has no Hamiltonian cycle, a contradiction to the premise.\\
\emph{Sufficiency:}
For any $3$-pyramid free 2-tree $G$ on more than two vertices, the unique Hamiltonian cycle of $G$ is obtained by using the edge set  $E^{'}=\{uv : |N_{G}(u) \cap N_{G}(v)|=1\}$.  $\hfill \qed$ \\
\end{proof}
\begin{theorem}\label{thm2}
Let G be a 2-tree.  For every $S\subset V(G)$, $c(G-S)\leq |S|$ if and 
only if G is $3$-pyramid free.
\end{theorem}
\begin{proof}
\emph{Necessity:}  Assume for a contradiction that $G$ has a $3$-pyramid.  This implies there exist $uv \in E(G)$ such that $|N_{G}(u) \cap N_{G}(v)|\ge3$.  By Lemma \ref{lem1}, $c(G-\{u,v\})\ge3$, a contradiction to the premise. \\
\emph{Sufficiency:} follows from Theorem \ref{thmchvatal} and \ref{thm1}. $\hfill \qed$
\end{proof}
\begin{corollary}
For a 2-tree $G$,  $G$ has a Hamiltonian cycle if and only if for every 
$S\subseteq V(G)$, $c(G-S)\leq |S|$.  
\end{corollary} 
Proof follows from Theorem \ref{thm1} and \ref{thm2}. $\hfill \qed$   \\
It is easy to see that graphs with Hamiltonian cycles contain  
Hamiltonian paths as well.  However, the converse is not true always.  Like 
Hamiltonian cycle problem there is no known necessary and sufficient 
condition for the existence of Hamiltonian paths in general graphs.  We 
below recall a necessary condition on graphs having Hamiltonian paths. 
\begin{lemma}\label{lem2}
Let $G$ be a connected graph.  If $G$ has a Hamiltonian path, then for every $S \subset V(G)$, $c(G-S) \leq |S| + 1$. 
\end{lemma}
\begin{proof} 
Suppose to the contrary assume that in $G-S$ there exist at least $|S|+2$ components. 
Any Hamilton path $P$ switches between different components at least $|S|+1$ times each time using a different element of  $S$, a contradiction.  $\hfill \qed$ 
\end{proof}
\begin{lemma}\label{lem3}
 Let $G$ be a 2-tree.  If $G$ contains a $4$-pyramid as an induced subgraph, then $G$ has no Hamiltonian path.
\end{lemma}
\begin{proof}
Let the $4$-pyramid in $G$ is due to the edge $uv$.  Clearly, $|N_{G}(u) \cap N_{G}(v)| \geq 4$. 
By Lemma \ref{lem1}, $c(G-\{u,v\}) >|\{u,v\}|+1$ and from  Lemma \ref{lem2}, it follows that $G$ has no Hamiltonian path.  $\hfill  \qed$ 
\end{proof}
The converse of the above lemma is not true and a counter example is illustrated in Figure \ref{fig3}. 
The example highlights the fact that there exist 2-trees with no $4$-pyramid and contain $3$-pyramids, yet it does not have Hamiltonian paths.  We shall now focus our structural analysis on 2-trees containing  $3$-pyramids.  In Lemma $\ref{lem4}$,  we show that $4$-pyramid free 2-trees having exactly one $3$-pyramid has a Hamiltonian path.
\begin{figure}[h!]
\begin{center}
	\includegraphics[scale=0.9]{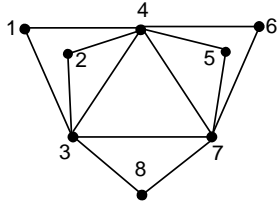}
\caption{$4$-pyramid free \emph{2-}tree having no Hamiltonian path}
\label{fig3}
\end{center}
\end{figure} 
\begin{lemma}\label{lem4}
 Let $G$ be a $4$-pyramid free 2-tree.  If $G$ contains exactly one $3$-pyramid as an induced subgraph, then there exist a Hamiltonian path in $G$.
\end{lemma}
\begin{proof}
Let the $3$-pyramid is on the edge $uv$.  Note that $|N_{G}(u) \cap N_{G}(v)| = 3$.  By Lemma \ref{lem1}, $c(G-\{u,v\})=3$ and let  $C_{1}, C_{2},$ and $C_{3}$ be those components.  Let $V(C_{1}) = \{u_{1}, u_{2}, \ldots, u_{i}\}, i \geq 1$,  $V(C_{2}) = \{ v_{1}, v_{2}, \ldots, v_{j}\}, j \geq 1$, and $V(C_{3}) = \{w_{1}, w_{2}, \ldots, w_{k}\}, k \geq 1$.  Consider the graphs $G_{i}, 1\le i\le3$ induced on $V(C_{i}) \cup \{u,v\}, 1\le i\le3$, respectively.  Clearly each $G_{i}, 1\leq i\leq 3$ is a 3-pyramid free 2-tree.  By Theorem 1, each $G_{i}, 1\leq i\leq 3$ has a Hamiltonian cycle and hence a Hamiltonian path.  The $(u,v)$-Hamiltonian path of $G_{1},  G_{2},$ and  $G_{3}$ are $(u, u_{1}, u_{2}, \ldots, u_{i}, v)$ ,  $(u, v_{1}, v_{2}, \ldots, v_{j},v)$, and $(u, w_{1}, w_{2}, \ldots, w_{k}, v)$, respectively.  The  path $(u_{1}, u_{2}, \ldots, u_{i}, v, v_{j},$ $ v_{j-1}, \ldots, v_{2}, v_{1}, u, w_{1}, w_{2}, \ldots, w_{k})$ is a Hamiltonian path in $G$. $\hfill \qed$
\end{proof}
We next present some combinatorial observations on $4$-pyramid free 2-trees with at least two $3$-pyramids for the existence of Hamiltonian paths. We also observe that not all such 2-trees possess Hamiltonian paths.  From now on we shall work with such 2-trees for our discussion.
\subsection{A Simplification (Vertex Pruning)} \label{secg0}
We now present an approach that transforms a $4$-pyramid free 2-tree with $3$-pyramids into a 2-tree without $3$-pyramids.  Intuitively, for such a $3$-pyramid with base edge $uv$, there are three 2-trees growing out of $uv$.   While pruning, out of the three 2-trees we retain two and prune the other.  While doing so, to remember the pruned 2-tree, we introduce coloring and labeling as part of our approach.  Coloring of $uv$ signifies that there is a 2-tree $H$ growing from $uv$ and $label(uv)$ signifies the vertices of $H$.\\\\
For a 2-tree $G$, by vertex pruning we remove vertices of degree 2 satisfying some property and color some of the edges in $G$, based on the \emph{closeness} property.  In particular, a vertex $v$ of degree 2 is pruned if its close edge, $close(v)$ is not colored and on pruning $v$, $close(v)$ is colored blue.   
Let $N_G(v)=\{u,w\}$, and on deleting $v$, we color the vertices $u,w$ blue and also the edge $uw$ blue.  
We remember the pruned vertices using a label associated with $uv$.  Initially all the edges are unlabeled, 
i.e., $label(uv)=\epsilon$ (empty string) for every $uv\in E(G)$.
On deleting $v$, we label $uw$ as follows:
\begin{itemize}
\item if $label(uv)=label(vw)=\epsilon$, then $label(uw)=(v)$
\item if $label(uv)=\epsilon$, and $label(vw)\neq\epsilon$ then $label(uw)=(v,label(vw))$
\item if $label(uv)\neq\epsilon$, and $label(vw)=\epsilon$ then $label(uw)=(label(uv),v)$
\item otherwise $label(uw)=(label(uv),v,label(vw))$
\end{itemize} 
For example, if the blue edges $uv$ and  $vw$ are labeled $ (u_{1},u_{2},\ldots,u_{i} )$ and $ (v_{1},v_{2},\ldots,v_{j})$,  respectively, then the label of the new blue edge  $uw$  will be $ (u_{1},u_{2},\ldots,u_{i},v,v_{1},v_{2},\ldots,v_{j}) $.  \\\\
For any 2-tree $G$, we define a sub 2-tree $G^0$ of $G$, which is obtained by recursively pruning 2-degree vertices $v$ of $G$ such that $close(v)$ is uncolored.  Note that if $G$ is a $4$-pyramid free 2-tree, then $G^0$ is $3$-pyramid free.
Further, for every 2-degree vertex $v$ in $G^{0}$, $close(v)$ is blue.  Since $G^{0}$ is $3$-pyramid free, $G^{0}$ contains a Hamiltonian cycle, and hence a Hamiltonian path as well.  However, our objective is to find a Hamiltonian path in $G^{0}$ containing all the blue edges, as labels of blue edges records the pruned vertices.  Further, such a Hamiltonian path can be easily extended to a Hamiltonian path in $G$ using the labels.  Given this observation, we would like to investigate $G^0$ to get some more insights.  We call $G^0$ as the \emph{vertex pruned} 2-tree of $G$.  An \emph{expanded} 2-tree of $G^0$ is a 2-tree obtained by growing each blue edge in $G^0$ with a $3$-pyramid free 2-tree corresponding to the label of the blue edge. 
We define the \emph{Blue graph} $B(G^0)$ of $G^0$ as a sub graph induced on the blue edges of $G^0$.  
For the next lemma we consider a $4$-pyramid free 2-tree $G$ with at least two $3$-pyramids and let $G^0$ be the vertex pruned 2-tree of $G$ and $B(G^0)$ be the blue graph of $G^0$.
\begin{figure}[h!]
\begin{center}
		\includegraphics[scale=1.1]{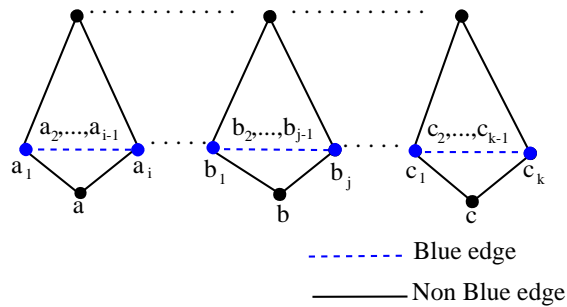}
\caption{Preprocessed 2-tree $G^{0}$ having three $2$-degree vertices.}
\label{figlem5}
\end{center}
\end{figure}
\newpage
\begin{lemma}\label{lem5}
If $G$ has a Hamiltonian path, then the following hold: \\
(i) $G^{0}$  has exactly two vertices of degree $2$.\\
(ii)  $\Delta (B(G^{0})) \le4$ \\
(iii) For $s\in V(G^0)$ such that $d_{G^0}(s)=2$ and $N_{G^0}(s)=\{u,v\}$, at most one of $u,v$ has degree $3$ in $B(G^0)$
\end{lemma}
\begin{proof}
\emph{(i)} Clearly, there are at least two vertices of degree $2$ in $G^{0}$ as $G$ has at least two $3$-pyramids.  Assume for a contradiction that there exist at least three vertices $a,b,c$ of degree $2$ in $G^{0}$ such that $N_{G^0}(a)=\{a_1,a_i\}$, $N_{G^0}(b)=\{b_1,b_j\}$, $N_{G^0}(c)=\{c_1,c_k\}$ (see Figure \ref{figlem5}).  
Clearly, $G^0$ has a Hamiltonian cycle and hence a Hamiltonian path.  Now we claim that any longest path $P$ in $G^0$ (which is a Hamiltonian path in $G^0$) cannot be transformed to a Hamiltonian path $P'$ in $G$.  Note that one of $a,b,$ and $c$ has a specific order of appearence in $P$.  In particular,  either $a_1,a,a_{i}$ or $b_1,b,b_{j}$ or $c_1,c,c_{k}$ appear consecutively in $P$.  Without loss of generality, let $b_1,b,b_{j}$ appear consecutively in $P$.   While extending $P$ to $P'$, we must include $label(b_1,b_{j})$, thus we get $(\ldots,b_1,label(b_1b_{j}),b_{j},\ldots)$.  However in this extension $b$ is unvisited in $P'$.  Therefore, $P'$ is not a Hamiltonian path.  This shows that any $P$ can not be extended to any Hamiltonian path in $G$, a contradiction to the premise. \\\\
\emph{(ii)} If $\Delta (B(G^{0}))\ge5$, then there exist $\{v,u_1,\ldots,u_5\}\subseteq V(B(G^0))$ such that $vu_i\in E(B(G^0)), 1\le i\le 5$.  Since $G^0$ is $3$-pyramid free and contains a Hamiltonian cycle, $G^0-\{v\}$ is connected.  Further, there exist a path $u_1,\ldots,u_2,\ldots,u_3,\ldots,u_4,\ldots,u_5$ in $G^0-\{v\}$ and for every $1<i<5$, $u_{i-1}$ and $u_{i+1}$ are in different components of $G^0-\{vu_i\}$.  Note that $c(G-\{u_2,v,u_4\})>4$, and by Lemma \ref{lem2}, $G$ has no Hamiltonian path, a contradiction.  
\\\\
\emph{(iii)} Assume for a contradiction that $d_{B(G^{0})}(u)=d_{B(G^{0})}(v)=3$.  Note that the edge $uv$ is a blue edge.  
If there exist a blue edge $uz', z'\neq v$, such that $c(G^0-\{u,z'\})=2$, then $c(G-\{u,v,z'\})>4$, and by Lemma \ref{lem2}, $G$ has no Hamiltonian path, a contradiction.  Hence we can assume that there exist two edges $us,uy$ such that $c(G^0-\{u,s\})<2$, $c(G^0-\{u,y\})<2$, and $us,uy$ are blue as shown in Figure \ref{figlem5.3}.  Symmetric argument holds for the vertex $v$.   We now show by case analysis that any longest path $P$ in $G^0$ can not be transformed into any Hamiltonian path $P'$ in $G$.  Since $P$ is a Hamiltonian path, $P$ must contain the vertex $s$.  Depending on the position of $s$ in $P$ we see various possibilities $P_1,\ldots,P_5$ for $P$ as follows.  $P_1=(\ldots,u,s,v,\ldots)$, $P_2=(u,s,v,\ldots)$, $P_3=(v,s,u,\ldots)$, $P_4=(s,u,v,\ldots)$, $P_5=(s,v,u,\ldots)$.  Now we shall show that each of the above Hamiltonian paths $P_i,1\le i\le5$ in $G^0$ can not be extended to any Hamiltonian path $P'$ in $G$.  
\begin{figure}[h!]
\begin{center}
		\includegraphics[scale=1.2]{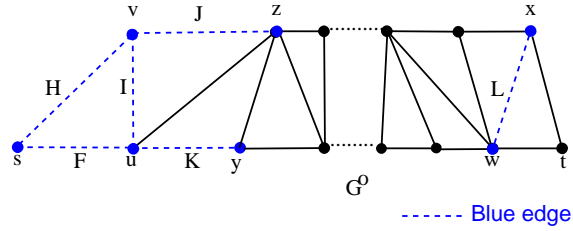}
\caption{An illustration for the proof of Lemma \ref{lem5}.(iii) }
\label{figlem5.3}
\end{center}
\end{figure}
We present the detailed case analysis in Table \ref{tabdelta3}. $\hfill \qed$
\begin{table}
\begin{center} 
\begin{tabular}{|l|l|}
\hline $~$Case & $~$Justification\\
\hline
Expanding $P_1$ to $P'$ & Since $P_1$ neither starts or ends in $u$ or $v$, the vertices in $I$ remain unvisited in $P'$\\\hline
Expanding $P_2$ to $P'$ & Since $P_2$ starts or ends at $u$, on expanding, one among $I,K,L$ or the vertex $t$ remain unvisited in $P'$.  \\
& if $P'$ visits the vertices $(I,u,F,s,H,v,J,z,\ldots,L,w,t)$, then the vertices in $K$ remain unvisited in $P'$. \\  
& if $P'$ visits the vertices $(K,u,F,s,H,v,J,z,\ldots,L,w,t)$, then the vertices in $I$ remain unvisited in $P'$. \\  
& if $P'$ visits the vertices $(I,u,F,s,H,v,J,z,\ldots,L,\ldots,K)$, then the vertex $t$ remain unvisited in $P'$. \\  
& if $P'$ visits the vertices $(I,u,F,s,H,v,J,z,\ldots,t,\ldots,K)$, then the vertices in $L$ remain unvisited in $P'$. \\  \hline
Expanding $P_3$ to $P'$ & Since $P_3$ starts or ends at $v$, on expanding, one among $I,J,L$ or the vertex $t$ remain unvisited in $P'$.  \\
& Arguments are symmetric to that of $P_2$\\\hline
Expanding $P_4$ to $P'$ & On expanding, one among $K,L$ or the vertex $t$ remain unvisited in $P'$.  \\
& if $P'$ visits the vertices $(H,s,F,u,I,v,J,z\ldots,L\ldots K)$, then the vertex $t$ remain unvisited in $P'$. \\  
& if $P'$ visits the vertices $(H,s,F,u,I,v,J,z\ldots,t\ldots K)$, then the vertices in $L$ remain unvisited in $P'$. \\  
& if $P'$ visits the vertices $(H,s,F,u,I,v,J,z\ldots,L,w,t)$, then the vertices in $K$ remain unvisited in $P'$. \\  \hline
Expanding $P_5$ to $P'$ & On expanding, one among $J,L$ or the vertex $t$ remain unvisited in $P'$.  \\
& Arguments are symmetric to that of $P_4$\\ \hline
\end{tabular}
\end{center}
\caption{Possibilities of expanding $P$ in Lemma \ref{lem5}.(iii)}
\label{tabdelta3}
\end{table}
\end{proof} 
\subsection{Another Simplification} \label{secg1}
In the previous section we have investigated the structure of $4$-pyramid free 2-trees with at least two $3$-pyramids by introducing the notion vertex pruning.  In this section, we shall obtain some more insights by introducing another simplification.  Our definition of $G$ and $G^0$ remains the same and in this section, we do not work with arbitrary $G^0$, instead, we work with $G^0$ satisfying the following conditions to obtain the 2-tree $G^1$. \\
(i) $G^{0}$  has exactly two vertices of degree $2$.\\
(ii)  $\Delta (B(G^{0})) \le4$ \\
(iii) For $s\in V(G^0)$ such that $d_{G^0}(s)=2$ and $N_{G^0}(s)=\{u,v\}$, at most one of $u,v$ has degree $3$ in $B(G^0)$ \\
Note that this is precisely the conclusion of Lemma \ref{lem5}.
The results presented in this section are based on such restricted $G^0$ and its corresponding $G^1$. 
We define the 2-tree $G^1$ obtained from such $G^0$ as follows;  let $s,t$ be two vertices of degree $2$ in $G^{0}$ and $G^{1}$=$G^{0}-\{s,t\}$, and $N_{G^0}(s)=\{u,v\}$ and  $N_{G^0}(t)=\{x,w\}$.  
We shall classify four types of Hamiltonian paths in $G^1$ based on $d_{B(G^1)}(z), z\in \{N_{G^0}(s)\cup N_{G^0}(t)\}$.  \\\\
Type 1 $(u,x)$-Hamiltonian path if $d_{B(G^1)}(v)=2$, $d_{B(G^1)}(w)=2$, $d_{B(G^1)}(u)=1$, and $d_{B(G^1)}(x)=1$.\\
Type 2 $(u,x)$-Hamiltonian path if $d_{B(G^1)}(z)=1, z\in\{u,v,w,x\}$ and $d_{G^1}(u)=d_{G^1}(x)=2$.\\
Type 3 $(u,x)$-Hamiltonian path if $d_{B(G^1)}(v)=2$, $d_{B(G^1)}(u)=1$ and $d_{B(G^1)}(z)=1, z\in\{w,x\}$ and $d_{G^1}(x)=2$.\\
Type 4 $(u,x)$-Hamiltonian path if $d_{B(G^1)}(w)=2$, $d_{B(G^1)}(x)=1$ and $d_{B(G^1)}(z)=1, z\in\{u,v\}$ and $d_{G^1}(u)=2$.\\
\begin{theorem}\label{thm3}
 $G$ has a Hamiltonian path if and only if $G^1$  has type 1 or type 2 or type 3 or type 4 $(u,x)$-Hamiltonian path containing all the blue edges of $G^{1}$. 
\end{theorem}
\begin{proof}
\emph{Sufficiency:} Let $R=(u,\ldots ,b_i,b_{i+1},\ldots,x)$ be a $(u,x)$-Hamiltonian path containing all the blue edges.  Replace every blue edge $b_ib_{i+1}$ with $(b_i,H,b_{i+1})$ where $H$ is the label of the blue edge in $G^{1}$ to get the expanded path $R'$.  When $R'$ is further extended by including labels, we get a path $P=(label(vs),s,label(su),R',label(xt),t,label(tw))$, which is a Hamiltonian path in $G$. \\
\emph{Necessity:} Let $b_ib_{i+1}$ is a blue edge in $G^1$.  Assume for a contradiction that there is no Hamiltonian path $Q$ in $G^1$ such that $\{b_i,b_{i+1}\}$ appear consecutively in $Q$. That is, $Q=(v,\ldots,b_i,\ldots,b_{i+1}\ldots,w)$ does not contain all blue edges of $G^1$.  On expanding $Q$, to get a path $P$ in $G$, clearly the $label(b_i,b_{i+1})$ does not appear in $P$.  This implies that $P$ is  not a Hamiltonian path in $G$, contradicting the premise.
$\hfill \qed$
\end{proof}
\begin{lemma}\label{lem9}
If $G$ has a Hamiltonian path, then $\Delta (B(G^{1})) \leq 2$. 
\end{lemma}
\begin{proof}
Assume for a contradiction that there exist a vertex $z$ such that $d_{B(G^1)}(z)\ge3$.  Since there exist a Hamiltonian path $P$ in $G$, by Theorem \ref{thm3}, there exist a $(u,x)$-Hamiltonian path $Q$ in $G^1$ containing all the blue edges of $G^1$.  Clearly, $Q$ must contain all three blue edges incident on $z$.  However, it is well known that a path can not contain three edges having a vertex in common, a contradiction.  Therefore, no such $z$ exists.  $\hfill  \qed $
\end{proof}
\begin{lemma}\label{lem8}
If $G$ has a  Hamiltonian path, then for every vertex $s\in V(G^0)$ such that $d_{G^0}(s)=2$, $N_{G^0}(s)=\{u,v\}$, one of the following holds:\\
(1) $d_{B(G^1)}(u)=1$,  $d_{B(G^1)}(v)=2$\\
(2) $d_{B(G^1)}(u)=2$,  $d_{B(G^1)}(v)=1$\\
(3) $d_{B(G^1)}(u)=d_{B(G^1)}(v)=1$
\end{lemma}
\begin{proof} 
Note that since the edge $uv$ is blue, $d_{B(G^1)}(z)\ge1, z\in \{u,v\}$.  From Lemma \ref{lem9} it follows that $d_{B(G^1)}(z)\le2$.  So $1\le d_{B(G^1)}(z)\le2$.  We now show that $d_{B(G^{1})}(u)=d_{B(G^{1})}(v)=2$ is not possible.  Assume for a contradiction, $d_{B(G^{1})}(u)=d_{B(G^{1})}(v)=2$. 
The proof of this claim is similar to Lemma \ref{lem5}.(iii) with minor modification on technical details.  
If there exist a blue edge $uz', z'\neq v$, such that $c(G^1-\{u,z'\})=2$, then $c(G^0-\{u,z'\})=2$, and further
$c(G-\{u,v,z'\})>4$.  By Lemma \ref{lem2}, $G$ has no Hamiltonian path, a contradiction.  Hence we can assume that there exist an edge $uy$ such that $c(G^0-\{u,y\})<2$, and $uy$ is blue as shown in Figure \ref{figlem8}.  Also, $c(G^1-\{u,y\})<2$, and $uy$ is blue in $G^1$.  Any longest path $Q$ in $G^1$ (which is also a Hamiltonian path in $G^1$) must contain $\{u,v\}$.  That is, $Q$ can be one of $Q_1=(v,u,\ldots)$, $Q_2=(u,v,\ldots)$, $Q_3=(\ldots,u,v,\ldots)$.  Since $G^0$ has a Hamiltonian path, when $Q$ is extended to a Hamiltonian path $P$ in $G^0$, it will include vertices $s$ and $t$.  On such expansion path $P$ will give one of $P_1$ to $P_5$ mentioned in Lemma \ref{lem5}.  In particular $Q_1$ is expanded to $P_i$, $i\in \{3,5\}$, $Q_2$ to $P_j$, $j\in \{2,4\}$, and $Q_3$ to $P_1$. 
At this point an anlaysis similar to Lemma \ref{lem5} will establish that $P_1$ to $P_5$ can not be extended to any Hamiltonian path $P'$ in $G$.  This shows that $Q$ can not be extended to any Hamiltonian path $P'$ in $G$, a contradiction. $\hfill\qed$
\begin{figure}[h!]
\begin{center}
		\includegraphics[scale=1.2]{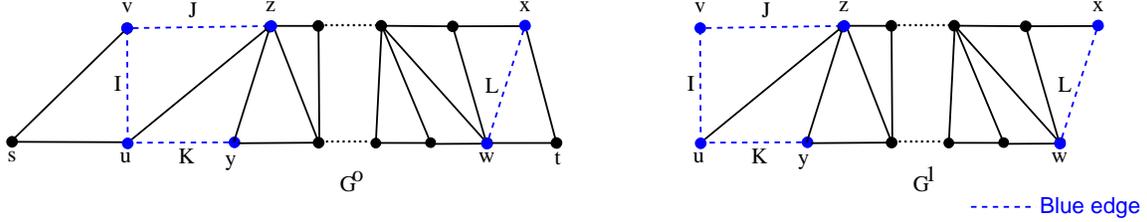}
\caption{An illustration for the proof of Lemma \ref{lem8} }
\label{figlem8}
\end{center}
\end{figure}
\end{proof}
Although in Lemma \ref{lem5} we have shown $\Delta(B(G^0))$ is at most $4$, there are exactly four 2-trees for which $\Delta(B(G^0))=4$.  For the rest $\Delta(B(G^0))$ is at most $3$ which we shall prove in the next lemma.\\\\
To present the next lemma we fix the following notation.  We define four special $3$-pyramid free 2-trees, $H_1,H_2,H_3,$ and $H_4$ as follows.  $V(H_i)=\{v,u,w,x,y\},$ and $E(H_i)=\{vu,vw,vx,vy,uw,wx,xy\},1\le i\le4$.  All the edges incident on $v$ are blue for each $H_i,1\le i\le4$.  Additionaly, the edge $uw$ is blue in $H_2$, the edge $xy$ is blue in $H_3$, and the edges $uw,xy$ are blue in $H_4$.  Note that each $H_i, 1\le i\le4$ is a $G^0$ for some $G$.
\begin{lemma} \label{lem6}
If $G$ has a Hamiltonian path, and  $\Delta (B(G^{0}))=4$, then $G^0\in \mathcal{H}=\{H_1,H_2,H_3,H_4\}$. 
\end{lemma}
\begin{proof}
Assume for a contradiction that there exist $H_0\notin \mathcal{H}$.  If $|V(H_0)|=5$, then the structure of $H_0$ is similar to $H_1$ and the only difference is the edge $wx$ is blue in $H_0$.  Then note that $c(G-\{v,w,x\})>4$ and by Lemma \ref{lem2}, $G$ has no Hamiltonian path, a contradiction.  Therefore, $|V(H_0)|>5$.  Let $V(H_0)=\{v,u,w,x,y,u_1,\ldots,u_k\}, k\ge1$, $\{vu,vw,vx,vy\}\subset E(H_0)$, and the edges $\{vu,vw,vx,vy\}$ are blue.  Clearly, there exist paths $P_{uw},P_{wx},P_{xy}$ in $H_0-\{v\}$.  If $P_{wx}$ is not an edge in $H_0-\{v\}$ or $wx$ is a blue edge, then $c(G-\{v,w,x\})>4$, again a contradiction.  Therefore, $P_{wx}$ must be an edge and $wx$ is not blue.  Now we shall see the adjacency of vertices $u_i, 1\le i\le k$.  Clearly, there is no $u_i, 1\le i\le k$ such that $\{v,w\}\subset N_{H_0}(u_i)$ or $\{v,x\}\subset N_{H_0}(u_i)$ or $\{w,x\}\subset N_{H_0}(u_i)$.  Existence of such $u_i$ yields $4$-pyramid in the former and $c(G-\{v,w,x\})>4$ in the later, a contradiction.  If $\{v,y\}\subset N_{H_0}(u_i)$ or $\{v,u\}\subset N_{H_0}(u_i)$, then either $c(G-\{v,y,w\})>4$ or $c(G-\{v,u,x\})>4$, a contradiction.  To complete the proof we shall focus on $P_{uw}$.  Let $P_{uw}=(u=w_1,\ldots,w_k=w), k\ge2$.  Note that $vw_i$ is not blue for $2\le i\le k-1$.  
Further, $d_{H_0}(u)\ge3$ and by Lemma \ref{lem5}.(i), there exist exactly two vertices of degree 2 in $H_0$.
Let $z,z'\in V(H_0)$ such that $d_{H_0}(z)=d_{H_0}(z')=2$.  
Note that the 2-tree $G^1$ obtained from $H_0$ on removing $z,z'$ has $d_{B(G^1)}(v)=3$, a contradiction to Lemma \ref{lem9}.  Symmetric argument holds for the path $P_{xy}$, and this completes a proof.
 $\hfill\qed$   
\end{proof}
\subsection{Yet Another Simplification (Edge Pruning)} \label{secg2}
In Section \ref{secg0} we have introduced first level pruning with the help of coloring and labeling of edges.  This helps to record the pruned vertices and further we obtained nice structural results on the blue graph.  It is natural to ask whether the existence of Hamiltonian path in $G$ is guaranteed (necessary and sufficient condition) using the Hamiltonian path containing all blue edges of $G^0$.  Surprisingly, the answer is no.  However, using the second level pruning presented in Section \ref{secg1}, we can guarantee a Hamiltonian path in $G$ using a Hamiltonian path containing blue edges.  Having highlighted this, it is natural to prune unnecessary (not part of any Hamiltonian path) non-blue edges from $G^0$ ($G^1$), and this is the objective of this section.  \\\\
With the definition of $G,G^0,G^1$ as before we shall introduce the following notations with respect to $G^1$. 
We work with a unique PEO $(v_{1},v_{2},\ldots,v_{k})$ of $G^{1}$ such that $d_{G^1}(v_1)=d_{G^1}(v_k)=2$.
\begin{itemize}
		\item \emph{Separator edges} $E_{s}$ = $\{e=uv : c(G^1-\{u,v\})>1\}$
		\item \emph{Non-separator edges} $E_{ns}$ = $E(G^{1})$ $\backslash E_{s}$.
	\item The \emph{left non-separator edge} of a vertex $v_{j}$ with $d_{G^{1}}(v_{j})>2$ is \emph{left($v_{j}$)}$=v_{i}v_{j}$ such that $i<j$ and $v_{i}v_{j}\in E_{ns}$.
	\item The \emph{right non-separator edge} of a vertex $v_{j}$ with $d_{G^{1}}(v_{j})>2$ is \emph{right($v_{j}$)}$=v_{j}v_{k}$ such that $j<k$ and $v_{j}v_{k}\in E_{ns}$. 
	\item \emph{Star vertices }$V_{s}$ = $\{v\in V(G^{1})$ such that $d_{G^{1}}(v)\geq 5\}$
	\item A \emph{forced star} refers to a star vertex with the blue left non-separator edge. 
	If $v_i\in V_s$ is a forced star, then $v_j\in N_{G^1}(v_i)\cap V_s$ such that $v_{i}v_{j}\in E_{s}$, $i<j$ is also a forced star. 
	\item A \emph{double forced star} refers to a forced star vertex with the blue right non-separator edge.
	\item For a blue separator edge $uv_{j}$ incident on a star vertex $u$, we define \emph{left separator edge}, \emph{left($uv_{j}$)}=$uv_{h}$ such that there is no $uv_{i}, h<i<j$ where $uv_{h}$, $uv_{i}\in E_{s}$.
	\item Similarly, \emph{right separator edge}, \emph{right($uv_{j}$)}=$uv_{h}$ such that there is no $uv_{i}, h>i>j$ where $uv_{h}$, $uv_{i}\in E_{s}$.
\end{itemize}
\begin{figure}[h!]
\begin{center}
	\includegraphics[scale=1]{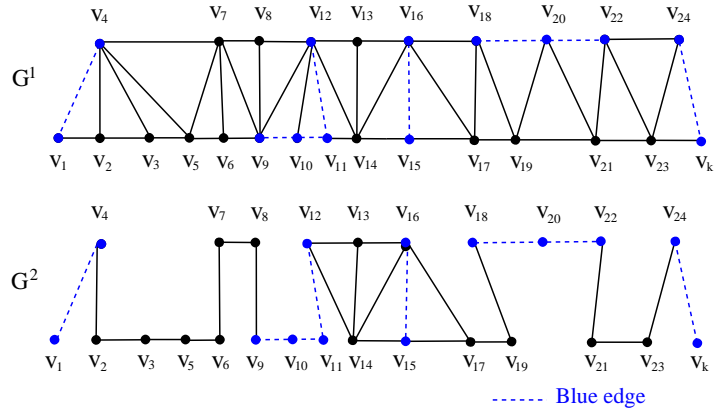}  
\caption{An illustration for edge pruning}
\label{g1}
\end{center}
\end{figure} 
With reference to Figure \ref{g1}, the left and right separator edges of the blue edge $v_{12}v_{11}$ of $G^{1}$ are $v_{12}v_{10}$ and $v_{12}v_{14}$, respectively.  The left and right non-separator edges of a star vertex $v_{12}$ are $v_{12}v_{8}$ and $v_{12}v_{13}$, respectively. 
\begin{obs}
For each $v\in V_{s}$, there exists at least three separator edges incident on $v$, and for each $u\in V(G^{1})$, there exist exactly two non-separator edges incident on $u$.
\end{obs}
As mentioned before, the objective of this section is to prune unnecessary non-blue edges in $G^1$ and towards this end, we define five sets of edges, $E_1,E_2,E_3,E_4,$ and $E_5$ (defined in Table \ref{tabedgesets}) whose removal from $G^1$ yields the graph $G^2$.  Since $E_5$ can not be empty, $G^2$ need not be a 2-tree.  In this section, we do not work with arbitrary $G^1$, instead, we work with $G^1$ satisfying the following conditions to obtain $G^2$.
\begin{enumerate}
\item  $\Delta (B(G^{1})) \leq 2$. 
\item For every vertex $s\in V(G^0)$ such that $d_{G^0}(s)=2$, $N_{G^0}(s)=\{u,v\}$, one of the following should hold
\begin{itemize}
\item $d_{B(G^1)}(u)=1$,  $d_{B(G^1)}(v)=2$ 
\item $d_{B(G^1)}(u)=2$,  $d_{B(G^1)}(v)=1$
\item $d_{B(G^1)}(u)=d_{B(G^1)}(v)=1$
\end{itemize}
\end{enumerate}
Note that this is precisely the conclusions of Lemmas \ref{lem9} and \ref{lem8}.  
\begin{table} [h!]
\begin{tabular}{|c|l|l|}
\hline
 Set & Definition & Intuitive justification $\backslash$ example \\
& & (see Figure \ref{g1}) \\\hline
$E_1$ & $\{yz : yz$ is not blue and $d_{B(G^{1})}(y)=2$ or $d_{B(G^{1})}(z)=2$ $\}$   & If there exist $d_{B(G^{1})}(y)=2$, then\\
&																																			& non-blue edges incident on $y$ are\\ 
&																																			 & not part of any $(v_1,v_k)$-Hamiltonian \\
&																																			 & path.  \\  
&																																			 & $E_1=\{v_{10}v_{12},v_{11}v_{14},v_{19}v_{20},v_{20}v_{21}\}$\\\hline
$E_2$ &  For $yv_j\in E_s$ and $yv_j$ is blue, $E_{a} = \{yv_{l} : l<j-1$ & $(v_1,v_k)$-Hamiltonian path either\\
&  or $l>j+1$ and $yv_{l}\in E_{s}$ and $yv_l$ is not blue $\}$.  & follows $(\ldots,v_{j-1},y,v_j,\ldots)$\\
& $E_{b}=\{left(y) : yv_{i}\in E_{a}, i<j-1$ and $left(y)$ is not blue$\}$   		&or $(\ldots,v_{j},y,v_{j+1},\ldots)$\\
& $E_{c}=\{right(y) : yv_{i}\in E_{a}, i>j+1$ and $right(y)$ is not blue$\}$  & $E_2=\{v_4v_3,v_4v_5,v_4v_7,v_{12}v_9,v_{12}v_8\}$ \\
& $E_2= E_{a} \cup E_b\cup E_c$  & \\\hline
$E_3$ & $E_{3}= \{ yz : $ there exist a maximal path $P_{v_{i}v_{j}}$ in $B(G^{1})$  such that &$(v_1,v_k)$-Hamiltonian path does not \\
& $v_pv_q\in (E(P_{v_{i}v_{j}})\cap E_s)\neq\emptyset$ and $|V(P_{v_{i}v_{j}})|>2$, $v_i\in C_u$ or $v_j\in C_x$ in  & contain $right(v_{i})$ or $left(v_{j})$ \\
& $G^1-\{v_p,v_q\}$ and $yz=right(v_{i})$ or $yz=left(v_{j})$ and $yz$ is not blue $\}$	 & $E_3=\{v_{12}v_8\}$\\\hline
$E_4$ &  Let $W_{1}$ and $W_{2}$ be two $(v_1,v_k)$-vertex disjoint paths in $G^{1}$, and $P_{v_iv_j}$ be a & $(v_1,v_k)$-Hamiltonian path either\\
&  maximal path in $B(G^{1})$ such that $E(P_{v_iv_j})\cap E_s=\emptyset$, $V(P_{v_iv_j})\cap V(W_{2})=\emptyset$, & follows $(\ldots,v_{p},P_{v_iv_j},\ldots)$\\
& $|V(P_{v_iv_j})|>2$ and $(N_{G^{1}}(v_i) \cap N_{G^{1}}(v_j)) \cap V(W_{2})=\emptyset$.  & or $(\ldots,P_{v_iv_j},v_q,\ldots)$  \\
& $E'_{ij}=$ $\{v_{p}v_q : p<q$, $v_{p}v_{q} \in E(W_{2})$, $v_{p}\in N_{G^{1}}(v_i)$, $v_{q}\notin N_{G^{1}}(v_i)$,  and $v_pv_q$ & $E_4=\{v_{19}v_{21}\}$\\
& is not blue $\}$ & \\
&$E''_{ij}= $ $\{v_{p}v_{q} : p<q$, $v_{p}v_{q}\in E(W_{2})$, $v_{q}\in N_{G^{1}}(v_j)$, $v_{p}\notin N_{G^{1}}(v_j)$, and $v_pv_q$ & \\
& is not blue $\}$  & \\
&$E_{4}= \bigcup\limits_{\forall i,j} E'_{ij}\cup \bigcup\limits_{\forall i,j}E''_{ij}$ &\\\hline
$E_5$ &  $E_5=\{e_1,\ldots,e_k\}$, where $e_i$ satisfies the property $\pi$ in $G^*=\widehat{G}-\{e_1,\ldots,e_{i-1}\}$, & $E_5=\{v_1v_2,v_9v_7,v_9v_6,v_7v_5,v_{18}v_{16},$\\
&  where $\widehat{G}=G^1-\bigcup\limits_{i=1}^{4}E_i$.  A non-blue edge  $e_i=pq$ in a block $D$ of $G^*$ is said to & $v_{18}v_{17},v_{22}v_{23},v_{22}v_{24},v_{23}v_k\}$\\
& satisfy the property $\pi$ if any one of the following holds: & \\
& (i) $\{p,r\}=N_{G^0}(z)$ where $d_{G^0}(z)=2$, $d_{B(G^1)}(r)=2$ and the vertices $p,r (\neq q)$ & \\
& are in block $D$. & \\
& (ii) $\{p,r\}=N_{G^0}(z)$ where $d_{G^0}(z)=2$, $d_{B(G^1)}(r)=d_{B(G^1)}(p)=1$,  $d_{G^1}(p)=2$ & \\
&  and the vertices $p,r (\neq q)$ are in block $D$. & \\
& (iii) $p$ is a cut vertex in $G^*$ and $pr\in E(G^*)$ where $d_{G^*}(r)=2$, and the vertices & \\
&  $p,r (\neq q)$ are in block $D$. & \\
& (iv) $p$ is a cut vertex in $G^*$ and $pr\in E(G^*)$ where $pr$ is a blue edge in $D$. & \\
& Further $r\neq q$& \\
\hline
\end{tabular}\\
\caption{Edges $E_1,\ldots,E_5$ definition and examples}
\label{tabedgesets}
\end{table}
The results presented in this section are based on such restricted $G^1$ and its corresponding $G^2$. 
An example illustrating the transformation is given in Figure \ref{g1}.  In Theorem \ref{thm4}, we establish a structural relation between $G^{1}$ and $G^{2}$.
\begin{theorem} \label{thm4}
There exist a type 1 or type 2 or type 3 or type 4 $(u,x)$-Hamiltonian path in $G^{1}$ containing all blue edges if and only if  there exist a $(u,x)$-Hamiltonian path in $G^{2}$ containing all blue edges. 
\end{theorem}
\begin{proof}  
The sufficiency is immediate as none of the blue edges are pruned for obtaining $G^2$.
For necessity; let $P_{ux}$ is $(u,x)$-Hamiltonian path of type $i, 1\le i\le3$ in $G^1$.
We shall now show that $P_{ux}$ is indeed $(u,x)$-Hamiltonian path in $G^2$ containing 
all blue edges.  The idea is to show that none of the edges in $E_1,\ldots,E_5$ is part of $P_{ux}$ in $G^1$.
For a contradiction, assume $P_{ux}$ has some edges from $E_1\cup\cdots\cup E_5$.
\begin{itemize}
\item
\emph{Case 1:}
$P_{ux}$ contains an edge $zl\in E_1$ in addition to the blue edges $hz,zj$ incident on $z$.  This shows that the path $P_{ux}$ has 3 edges sharing the vertex $z$ in common, a contradiction.
\item
\emph{Case 2:} 
$P_{ux}$ contains an edge $v_iv_l\in E_2$ in addition to the blue edge $v_iv_j\in E_s$.  As per the definition of $E_2$, $v_iv_{j-1}, v_iv_{j+1}\notin E_2$.  This implies that $l\ge j+2$ and $P_{ux}$ is of the form $(u,\ldots,v_j,v_i,v_l,\ldots,x)$ or  $l\le j-2$ and $P_{ux}$ is of the form $(u,\ldots,v_l,v_i,v_j,\ldots,x)$
as it has to contain the edge $v_iv_l$.  However, the vertex $v_{j+1}$ in the former or $v_{j-1}$ in the later is unvisited in $P_{ux}$, contradicting the fact that $P_{ux}$ is a Hamiltonian path. 
\item
\emph{Case 3:}  
Assume $P_{ux}$ contains an edge $v_jv_l\in E_3$.  Clearly, $P_{ux}$ contains $P_{v_iv_j}=(v_i,\ldots,v_{j'},v_j)$ (see $E_3$ in Table \ref{tabedgesets} for the definition of $P_{v_iv_j}$) as a sub path.  Further $P_{ux}$ contains the blue edge $v_jv_{j'}\in E_s$.  
Observe that $P_{ux}$ is of the form $(u,\ldots,P_{v_iv_j},v_l,\ldots,x)$ or $(u,\ldots,v_l,P_{v_jv_i},\ldots,x)$.  This implies that, $P_{v_lx}$ or $P_{v_ix}$ must contain one of  $\{v_j,v_{j'}\}$.  From the above argument it follows that $v_j$ or $v_{j'}$ appears more than once in $P_{ux}$, contradiction to the definition of Hamiltonian path.  A symmetric argument holds true for the other edges in $E_3$.
\item
\emph{Case 4:}
$P_{ux}$ contains an edge $v_pv_q\in E_4$.  Clearly, $P_{ux}$ contains $P_{v_iv_j}$ (see $E_4$ in Table \ref{tabedgesets} for the definition of $P_{v_iv_j}$) as a sub path. 
If $v_pv_q\in E'_{ij}$, then $P_{xu}$ is of the form $(x,\ldots,P_{v_jv_i},v_p,v_q,\ldots,u)$ or $(x,\ldots,v_q,v_p,P_{v_iv_j},\ldots,u)$.  This implies that $P_{v_qu}$ or $P_{v_ju}$ must contain one of  $\{v_i,v_{p}\}$.  From the above argument it follows that $v_i$ or $v_{p}$ appears more than once in $P_{xu}$, a contradiction.  
If $v_pv_q\in E''_{ij}$, then $P_{ux}$ is of the form $(u,\ldots,P_{v_iv_j},v_q,v_p,\ldots,x)$ or $(u,\ldots,v_p,v_q,P_{v_jv_i},\ldots,x)$.  This implies that $P_{v_px}$ or $P_{v_ix}$ must contain one of  $\{v_j,v_{q}\}$.  From the above argument it follows that $v_j$ or $v_{q}$ appears more than once in $P_{ux}$, a contradiction.  
\item
\emph{Case 5:}  
$P_{ux}$ contains an edge $pq\in E_5$.  We present case analysis and arrive at a contradiction in each of them.  See $E_5$ in Table \ref{tabedgesets} for conditions mentioned below. \\
\emph{case a:} Condition (i) or (ii) holds.  Note that $p$ is either $u$ or $x$ and clearly $pr$ is a blue edge.  It follows that $u$ or $x$ has two edges incident to it in $P_{ux}$.  However, any $(u,x)$-Hamiltonian path can not have two edges incident on $u$ or $x$ which are the end vertices of $P_{ux}$, a contradiction.\\
\emph{case b:} Condition (iii) holds.  Note that $P_{ux}$ is of the form $(u,\ldots,z,p,q,\ldots,x)$ or $(u,\ldots,q,p,z,\ldots,x)$ where $z$ is not in $D$.  Since $d_{G^*}(r)=2$, and $P_{ux}$ contains all the vertices in $G^1$, both the edges incident on $r$ are in $P_{ux}$.  In particular, $pr\in E(P_{ux})$, and therefore, there exist three edges in $P_{ux}$ which are incident on  the vertex $p$, a contradiction.\\
\emph{case c:} Condition (iv) holds.  Note that $P_{ux}$ is of the form $(u,\ldots,z,p,q,\ldots,x)$ or $(u,\ldots,q,p,z,\ldots,x)$ where $z$  is not in $D$.  Since $pr$ is a blue edge, and $P_{ux}$ contains all the blue edges, the vertex $p$ has three edges incident on it in $P_{ux}$, a contradiction.
$\hfill \qed$  \end{itemize} 
\end{proof}
It is important to highlight that the three transformations ($G\rightarrow G^0\rightarrow G^1\rightarrow G^2$) discussed so far results in a graph with all the blue edges and relatively a few non-blue edges.  It is now appropriate to identify yes instances of $G^2$ ($G^2$ with Hamiltonian paths containing all the blue edges).  Towards this end, we present two structural observations and using which we can characterize all $G^2$ having Hamiltonian paths containing all the blue edges.  Further, it helps in identifying Hamiltonian paths in $G$ as well. \\\\
\textbf{Conflicting Paths}\\
Let $P_{v_{i}v_{j}}$ and $P_{v_{y}v_{z}}$ be two maximal vertex disjoint paths having no separator edges in $B(G^{2})$, the graph induced on the blue edges of $G^2$.
Note that the indices of the vertices represents their ordering in $\sigma$.  The paths $P_{v_{i}v_{j}}$, $|P_{v_{i}v_{j}}|\ge2$ and $P_{v_{y}v_{z}}$, $|P_{v_{y}v_{z}}|\ge3$ are said to be \emph{conflicting} if $y<i<z$ and $v_{i}v_{z}\notin E(G^{2})$.  
\newpage
\begin{lemma} \label{lem10}
If $G$ has a Hamiltonian path, then $G^{2}$ has no conflicting paths. 
\end{lemma}
\begin{proof}
 Assume for a contradiction that there exist a pair of conflicting blue paths $P_{v_{y}v_{z}}$ and $P_{v_{i}v_{j}}$ in $G^{2}$ such that $y<i<z$ in $\sigma$.  From Theorems \ref{thm3} and \ref{thm4}, note that there exists a $(u,x)$-Hamiltonian path $P_{ux}$ in $G^2$ containing all the blue edges.  The path $P_{ux}$ is either of the form $(u,\ldots,P_{v_yv_z},\ldots,P_{v_jv_i},\ldots,x)$ or $(u,\ldots,P_{v_iv_j},\ldots,P_{v_zv_y},\ldots,x)$.  Since the sub paths $(v_i,\ldots,x)$ and $(v_y,\ldots,x)$ contains a vertex in $\{v_j,v_{z}\}$, either $v_j$ or $v_z$ is visited twice in $P_{ux}$, a contradiction.  $\hfill\qed$
\end{proof}
\begin{obs} \label{obs4}
From the definition of $G^1$, there is a blue edge incident on $u$ in $G^1$ and from the definition of $G^2$, $d_{B(G^2)}(u)=1$.  Further, all the non-blue edges incident on $u$ are not present in $G^2$.  It follows that, $d_{G^2}(u)=1$.  Similarly, $d_{G^2}(x)=1$.
\end{obs}
\begin{lemma} \label{lem11}
Let $U=u_1,\ldots,u_i$ be the cut vertices of $G^2$.  
If $G$ has a Hamiltonian path, then for each block $D$ of $G^{2}$ \\
(i) for every $u_j,j<i$ in $D$, $c(G^2-u_j)=2$.\\
(ii) $|D\cap U|\le 2$
\end{lemma}
\begin{proof}
Suppose there exist a vertex $u_j,j<i$ such that $c(G^2-u_j)>2$, then this contradicts Lemma \ref{lem2}.  
Therefore, for every $u_j$ in $D$, $c(G^2-u_j)=2$.  Suppose for a contradiction, the vertices $\{u_1,\ldots,u_j, j\ge3\}$ are present in  $D$.  Since $G$ has a Hamiltonian path, from Theorems \ref{thm3} and \ref{thm4}, there exists a $(u,x)$-Hamiltonian path $P_{ux}$ in $G^2$ containing all the blue edges.  
From Observation \ref{obs4}, $u\notin D$.  Clearly, $P_{ux}$ is of the form $(u,\ldots,u_1,\ldots,u_2,\ldots,u_3,\ldots,x)$, where $u_1$ is the first vertex of $D$ visited in $P_{ux}$, which is also in $U$. 
From the above observation note that the internal vertices of the sub path $P_{uu_1}$ are not in $D$, and the vertex $u_2$ is visited twice in the sub path $P_{u_1u_3}$, which is a contradiction.
$\hfill\qed$
\end{proof}
\begin{theorem} \label{thm5}
Let $U$ be the set of cut vertices in $G^2$.  There exist a $(u,x)$-Hamiltonian path containing all the blue edges in $G^{2}$ if and only if $G^2$ is connected and the following holds:\\
(i) For every $u_j\in U$, $c(G^2-u_j)=2$ and for each block $D$ of $G^{2}$, $|D\cap U|\le 2$.\\
(ii) $G^{2}$ has no double forced stars.  \\
(iii) $G^{2}$ has no conflicting paths.
\end{theorem}
\begin{figure}[h!]
\begin{center}
			\includegraphics[scale = 1.2]{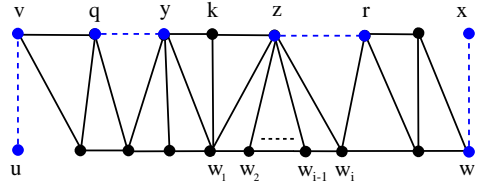}%
\end{center}
\caption{An illustration for the proof of Theorem \ref{thm5}}%
\label{figdf}%
\end{figure}
\begin{proof}
\emph{Necessity:} 
From Theorems \ref{thm3} and \ref{thm4}, we know that there exist a $(u,x)$-Hamiltonian path containing all the blue edges in $G^{2}$.  Further, using Lemma \ref{lem11}, condition (i) follows and using Lemma \ref{lem10}, condition (iii) follows.   
 For (ii), we assume on the contrary that there exist a double forced star in $G^{2}$.  Let $y$ be a forced star and $z$ be a double forced star as shown in Figure \ref{figdf} such that $q<y\le z<r$.  Since the vertex $z$ is a star vertex, in the neighborhood of $z$, apart from $k,r$, there exist vertices $w_1,\ldots,w_i,i\ge3$.  Since $G^2$ is connected, $d_{G^2}(w_{j})=3,1<j<i$.  Given that in $G^2$ there exist a $(u,x)$-Hamiltonian path $P_{ux}$ containing all the blue edges, the blue edges $qy,zr$ are in $P_{ux}$.  Clearly, $P_{ux}$ is of the form $(u,\ldots,q,y,\ldots,z,r,\ldots,x)$.  However, the vertices $\{w_2,\ldots,w_{i-1}\},i\ge3$ are not visited in $P_{ux}$, a contradiction to the fact that $P_{ux}$ is a Hamiltonian path in $G^2$.
\\\\\emph{Sufficiency:}  
Since $G^{2}$ is connected, from condition (i), there exist blocks $D_1,\ldots,D_k, k\ge2$ such that adjacent blocks $D_{j-1},D_j$ share a vertex $u_j\in U$, $1<j\le k$.  It follows from Observation \ref{obs4}, $u\in D_1, x\in D_k$ and by condition (i) every $D_j,2\le j\le k-1$, $|D_j\cap U|=2$.  
 To show that there exists a $(u,x)$-Hamiltonian path containing all the blue edges in $G^{2}$, it is sufficient to show that for every block $D_j,2\le j\le k-1$, there exists an $(p,q)$-Hamiltonian path containing all the blue edges of $D_j$, where $p,q\in D_j\cap U$.  We prove this claim using induction on the number of vertices of the block $D=D_j,2\le j\le k-1$.\\  
\emph{Base case:} $|D|=2$.  Clearly, $D$ contains the edge $pq$ with no double forced stars and no conflicting paths.  The edge $pq$ itself is an $(p,q)$-Hamiltonian path in $D$. \\ 
\emph{Induction hypothesis:}  Let $D$ be block with no double forced stars and no conflicting paths on less than $i, i\geq 3$ vertices and $p,q\in D\cap U$.  Assume that there exist an $(p,q)$-Hamiltonian path in $D$ containing all the blue edges. \\ 
\emph{Induction Step:} Let $D$ be a block in $G^{2}$ on $i$ vertices, $i\geq 3$ such that $D$ has no double forced stars and no conflicting paths.  Let $p,q\in D\cap U$.  We now claim that there exist at most one blue edge incident on $p$.  Suppose there are at least two blue edges incident on $p$.  Since $p$ is a cut vertex, there exist an edge $pi$, $i\notin D$.  Since $\Delta_{B(G^1)}\le2$, $pi$ is not blue, and therefore, by the definition of $E_1$, $pi\in E_1$.  A contradiction to the fact that $p$ is a cut vertex.  We shall complete the inductive proof by considering the following cases.  \\
\emph{case 1:}   There exist a neighbor $j$ of $p$ such that $pj$ is blue and for all other neighbors $l$, $pl$ is not blue.  \\
\emph{case 2:}   All neighbors of $p$ are not blue. \\
To get the inductive sub problem, we prune the edges based on the above cases.  For case 1, we prune all the edges $pl$ in $D$ except the edge $pj$.  For case 2, let $j$ be the least indexed neighbor such that $j<l$ in $\sigma$, where $l$ is any other neighbor of $p$.  In this case we retain $pj$ and prune all such $pl$.  Clearly, in either case analysis, the resultant graph is free from conflicting paths and double forced stars.  Note the degree of $p$ in the resultant graph is two.  Let $j=w_1,\ldots,w_k,k\ge2$ such that $w_i,w_{i+1}$ are adjacent and $w_k$ be the least indexed vertex such that degree of $w_k$ is at least 3.  Clearly, using the base case we know that there exist $(w_i,w_{i+1})$-Hamiltonian path, $1\le i\le k-1$ and by the induction hypothesis, we get $(w_k,q)$-Hamiltonian path containing all the blue edges.  Further, using these paths we get $(p,q)$-Hamiltonian path containing all the blue edges.  The induction is complete and the theorem follows.  $\hfill\qed$
\end{proof} 
\begin{theorem} \label{thm6}
A 2-tree $G$ contains a Hamiltonian path if and only if any one among the following holds:
\begin{enumerate}
	\item $G$ is $3$-pyramid free.
	\item $G$ is $4$-pyramid free and contains exactly one $3$-pyramid.
	\item $G$ is $4$-pyramid free and $G$ contains at least two $3$-pyramids and
				\begin{enumerate}
				\item $G^{2}$ is connected and for every $u_j\in U$, the set of cut vertices in $G^2$, $c(G^2-u_j)=2$ and for each block $D$ of $G^{2}$, $|D\cap U|\le 2$.
				\item $G^{2}$ is connected and $G^2$ has no double forced stars.  
				\item $G^{2}$ is connected and $G^2$ has no conflicting paths.				
				\end{enumerate}
\end{enumerate}
\end{theorem}
\begin{proof}
\emph{Necessity:} Since $G$ contains a Hamiltonian path, by Lemma \ref{lem3}, $G$ is $4$-pyramid free.  If $G$ is also $3$-pyramid free or contains exactly one $3$-pyramid, then the necessity follows for \emph{1} and \emph{2}.  If $G$ contains at least two $3$-pyramids, then using Theorems \ref{thm3}, \ref{thm4} and \ref{thm5}, the claim \emph{3} follows. This completes the necessity.  
\emph{Sufficiency:} 
From Theorems \ref{thm1}, \ref{thm3}, \ref{thm4}, \ref{thm5} and Lemma \ref{lem4}, $G$ contains a Hamiltonian path.  
This completes the sufficiency, and a proof of the theorem. $\hfill \qed$ 
\end{proof}
\newpage
\subsection{A polynomial-time algorithm to find a Hamiltonian path in a 2-tree}
\begin{algorithm}[h] 
\caption{Hamiltonian path in 2-trees:  {\em Hamiltonian\_Path(2-Tree G)}}
\label{Hamiltonianpath}
\begin{algorithmic}[1]
\STATE{If $G$ is $3$-pyramid free, then there exist a Hamiltonian path containing only non-separator edges.}
\STATE{If $G$ is $4$-pyramid free with exactly one $3$-pyramid, then starting from an appropriate vertex of the $3$-pyramid, there exist a Hamiltonian path containing only non-separator edges.}
\STATE{For $G$ with $4$-pyramid free with at least two $3$-pyramids, perform vertex pruning as defined in Section \ref{secg0} to get $G^0$.}
\STATE{If $G^0$ has at least three vertices of degree 2 or the maximum degree of $B(G^0)$ is at least 5, then print G has no Hamiltonian path.  Also check if there are exactly two vertices of degree 2, then its neighborhood contains at most one vertex of degree three in $B(G^0)$.}
\STATE{Find $G^1$ from $G^0$ followed by Sets $E_1$ to $E_5$ as defined in Section \ref{secg2} and prune them from $G^1$ to get $G^2$.}
\STATE{If the maximum degree of $B(G^1)$ is at least 3, then print G has no Hamiltonian path.  Also check for degree two vertex, its neighborhood contains at most one vertex of degree two in $B(G^1)$.}
\STATE{ If $G^{2}$ is not connected \textbf{or} has a cut vertex, $v$ such that $c(G^{2}-v)\geq 3$ \textbf{or} 
 there exist a block $D$ of $G^{2}$ that has more than two cut vertices in $G^{2}$ \textbf{or} 
 $G^{2}$ has a double forced stars \textbf{or} there exist conflicting paths in $G^{2}$, then print $G$ has no Hamiltonian path.}
\STATE{For each block $D$, find a spanning path in $G^2$ using Algorithm \ref{spanningpath} containing all the blue edges and expand the blue edges using the labels, and output the expanded path as a Hamiltonian path in $G$.}
\end{algorithmic}
\end{algorithm}
%
\begin{algorithm}[h!]
\caption{Spanning Path: \emph{Spanning\_Path($D$,$p$,$q$)}}
\label{spanningpath}
\begin{algorithmic}[1]
\STATE{If $D$ induces an edge then return the edge which is the trivial Hamiltonian path.}
\STATE{If there exist a blue edge $pj$ incident on $p$.  All other non-blue edges incident on $p$ are pruned.  }
\STATE{If none of the edges incident on $p$ are blue, then find the least indexed vertex $j$ in $\sigma$ among the neighbors of $p$.  Retain the edge $pj$ and prune all other edges incident on $p$.}
\STATE{After pruning, find the least indexed vertex $w_k$, in $p=w_1,\ldots,w_k,k\ge2$ such that $w_i,w_{i+1}$ are adjacent and degree of $w_k$ is at least 3.  Recursively find Spanning\_Path($D-\{w_1,\ldots,w_{k-1}\},w_k,q$).}
\end{algorithmic}
\end{algorithm}
\subsection{Proof of correctness and Run-time analysis}
Step 1,2,4 are correct due to Theorem 2, Lemmas 4 and 5.  Correctness of Step 6 is from Lemmas 6,7, and Step 7 is due to Theorems 4, 5, and 6.  Step 8 correctly produces the Hamiltonian path due to the Theorem \ref{thm5} and Algorithm \ref{spanningpath} is an implementation of the constructive proof mentioned in Theorem \ref{thm5}.  
\\\\ Using the standard Depth First Search tree, we can get the sets $E_1$ to $E_5$ and also the cut vertices and blocks of $G^2$.  A careful fine tuning of the DFS tree helps us to get spanning paths mentioned in Step 8 of our algorithm.  Therefore, the overall effort to output a Hamiltonian path is linear in the input size. 
\section*{Conclusions and Future Work}
In this paper, we have characterized the class of 2-trees having Hamiltonian paths.  Further, using our combinatorics, we have presented a polynomial-time algorithm to find Hamiltonian paths in 2-trees.  We believe that combinatorics presented here can be used in other combinatorial problems restricted to 2-trees.  A natural extension of our work is to look at an algorithm for finding Hamiltonian paths in $k$-trees, $k\geq 3$.  Also, algorithms presented here output just one Hamiltonian path if it exists in a given 2-tree.  A related problem is to generate all Hamiltonian paths in a given 2-tree.
\nocite{*}
\bibliographystyle{splncs2}
\bibliography{twotreeref}

\begin{thebibliography}{10}

\bibitem{bertossi}
A.A.Bertossi, M.A.Bonuccelli:
\newblock Hamiltonian circuits in interval graph generalizations.
\newblock Information Processing Letters,  195--200, (1986)

\bibitem{ainouche}
A.Ainouche:
\newblock Four sufficient conditions for hamiltonian graphs.
\newblock Discrete Mathematics,  195 -- 200, (1991)

\bibitem{app2}
A.Malakis:
\newblock Hamiltonian walks and polymer configurations.
\newblock Statistical Mechanics and its Applications Physica (A),  256--284,
  (1976)

\bibitem{pandapradhan}
B.S.Panda, D.Pradhan:
\newblock \textsc{NP}-completeness of hamiltonian cycle problem on rooted
  directed path graphs.
\newblock (preprint), (2008)

\bibitem{panda}
B.S.Panda, S.K.Das:
\newblock A linear time recognition algorithm for proper interval graphs.
\newblock Information Processing Letters,  153--161, (2003)

\bibitem{nash}
C.St.J.A.Nash-Williams:
\newblock On hamiltonian circuits in finite graphs, In: Proceedings of the
  American Mathematical Society.  466--467 (1966)

\bibitem{herz}
C.Thomassen:
\newblock Hypohamiltonian and hypotraceable graphs.
\newblock Discrete Mathematics, \textbf{9}(1),  91--96, (1974)

\bibitem{west}
D.B.West:
\newblock Introduction to graph theory 2$^{nd}$ Edition, (2003)

\bibitem{app1}
D.Dorninger:
\newblock Hamiltonian circuits determining the order of chromosomes.
\newblock Discrete Applied Mathematics,  159 -- 168, (1994)

\bibitem{gould1}
D.Duffus, R.J.Gould, M.S.Jacobson:
\newblock Forbidden subgraphs and the hamiltonian theme.
\newblock The Theory and Applications of Graphs, Wiley-Interscience(N.Y.),
  297--316, (1981)

\bibitem{newman}
D.J.Newman:
\newblock A problem in graph theory.
\newblock The American Mathematical Monthly,  611, (1958)

\bibitem{dirac}
G.A.Dirac:
\newblock Some theorems on abstract graphs.
\newblock Proc. London Math. Society,  69--81, (1952)

\bibitem{chatrand}
G.Chartrand, R.J.Gould, S.F.Kapoor:
\newblock On homogeneously traceable nonhamiltonian graphs.
\newblock Annals of the N. Y. Acad. of Sci.,  130--135, (1979)

\bibitem{app3}
G.Irina, O.Halskau, G.Laporte, M.Vlcek:
\newblock General solutions to the single vehicle routing problem with pickups
  and deliveries.
\newblock Euro. Jour. of Operational Research,  568--584, (2007)

\bibitem{giri}
G.Narasimhan:
\newblock A note on the hamiltonian circuit problem on directed path graphs.
\newblock Information Processing Letters,  167--170, (1989)

\bibitem{keirstead}
H.A.Kierstead, G.N.Sarkozy, S.M.Selkow:
\newblock On k-ordered hamiltonian graphs.
\newblock Journal of Graph Theory,  17--25, (1999)

\bibitem{s1}
H.J.Broersma:
\newblock On some intriguing problems in hamiltonian graph theory - a survey.
\newblock Discrete Mathematics,  47--69, (2002)

\bibitem{boersma}
H.J.Broersma, H.J.Veldman:
\newblock Restrictions on induced subgraphs ensuring hamiltonicity or
  pancyclicity of $k_{1,3}$-free graphs.
\newblock Contemporary Methods in Graph Theory, Mannheim-Wien-Zurich,
  181--194, (1990)

\bibitem{muller}
H.Muller:
\newblock Hamiltonian circuits in chordal bipartite graphs.
\newblock Discrete Mathematics,  291--298, (1996)

\bibitem{bondy}
J.A.Bondy, U.S.R.Murty:
\newblock Graph Theory with Applications.
\newblock Macmillan, London, (1976)

\bibitem{bondychvatal}
J.A.Bondy, V.Chv$\acute{a}$tal:
\newblock A method in graph theory.
\newblock Discrete Mathematics,  111--135, (1976)

\bibitem{Ryjacek}
J.Brousek, Z.Ryjacek, O.Favaron:
\newblock Forbidden subgraphs, hamiltonicity and closure in claw-free graphs.
\newblock Discrete Mathematics,  29--50, (1999)

\bibitem{keil}
J.M.Keil:
\newblock Finding hamiltonian circuits in interval graphs.
\newblock Information Processing Letters,  201--206, (1985)

\bibitem{collins}
K.L.Collins, L.B.Krompart:
\newblock The number of hamiltonian paths in a rectangular grid.
\newblock Discrete Mathematics,  29--38, (1997)

\bibitem{ibarra}
L.Ibarra:
\newblock A simple algorithm to find hamiltonian cycles in proper interval
  graphs.
\newblock Information Processing Letters,  1105 -- 1108, (2009)

\bibitem{schultz}
L.Ng, M.Schultz:
\newblock k-ordered hamiltonian graphs.
\newblock Journal of Graph Theory,  45--57, (1997)

\bibitem{golumbic}
M.C.Golumbic:
\newblock Algorithmic Graph Theory and Perfect Graphs.
\newblock Academic Press, New York, (1980)

\bibitem{tarjanplanar}
M.R.Garey, D.S.Johnson, R.E.Tarjan:
\newblock Planar hamiltonian circuit problem is \textsc{NP}-complete.
\newblock SIAM Journal on Computing,  704--714, (1976)

\bibitem{ore}
O.Ore:
\newblock Note on hamilton circuits.
\newblock The American Mathematical Monthly, ~55, (1960)

\bibitem{bedrossian}
P.Bedrossian:
\newblock Forbidden subgraph and minimum degree conditions for hamiltonicity.
\newblock Ph.D. Thesis, Memphis State University, (1991)

\bibitem{gould2}
R.J.Faudree, R.J.Gould:
\newblock Characterizing forbidden pairs for hamiltonian properties.
\newblock Discrete Mathematics,  45--60, (1997)

\bibitem{faudree2}
R.J.Faudree, R.J.Gould, A.V.Kostochka, L.Lesniak, I.Schiermeyer, A.Saito:
\newblock Degree conditions for k-ordered hamiltonian graphs.
\newblock Journal of Graph Theory,  199--210, (2003)

\bibitem{Faudree}
R.J.Faudree, R.J.Gould, M.S.Jacobson, L.Lesniak:
\newblock Neighborhood unions and highly hamiltonian graphs.
\newblock Journal of Graph Theory,  29--38, (1991)

\bibitem{faudree1}
R.J.Faudree, R.J.Gould, M.S.Jacobson, L.Lesniak:
\newblock On k-ordered graphs.
\newblock Journal of Graph Theory,  69--82, (2000)

\bibitem{Gould}
R.J.Faudree, R.J.Gould, M.S.Jacobson, R.H.Schelp:
\newblock Neighborhood unions and hamiltonian properties in graphs.
\newblock Journal of Combinatorial Theory Series B,  1--9, (1989)

\bibitem{s2}
R.J.Gould:
\newblock Updating the hamiltonian problem - a survey.
\newblock Journal of Graph Theory,  121--157, (1991)

\bibitem{s3}
R.J.Gould:
\newblock Advances on the hamiltonian problem - a survey.
\newblock Graphs and Combinatorics,  7--52, (2003)

\bibitem{gould3}
R.J.Gould, M.S.Jacobson:
\newblock Forbidden subgraphs and hamiltonian properties of graphs.
\newblock Discrete Mathematics,  189--196, (1982)

\bibitem{karp}
R.M.Karp:
\newblock Reducibility among combinatorial problems, In: Proc. of a Symposium
  on the Complexity of Computer Computations.  85--103 (1972)

\bibitem{hungdis}
R.W.Hung, M.S.Chang:
\newblock Linear-time algorithms for the hamiltonian problems on
  distance-hereditary graphs.
\newblock Theoretical Computer Science,  411--440, (2005)

\bibitem{hungint}
R.W.Hung, M.S.Chang:
\newblock Linear-time certifying algorithms for the path cover and hamiltonian
  cycle problems on interval graphs.
\newblock Applied Mathematics Letters,  648--652, (2011)

\bibitem{hungcir}
R.W.Hung, M.S.Chang, C.H.Laio:
\newblock The hamiltonian cycle problem on circular-arc graphs, In: Proc. of
  the Intl. Conf. of Engg. and Computer Scientists (IMECS, Hong Kong).  18--20
  (2009)

\bibitem{courcelle}
S.Arnborg, B.Courcelle, A.Proskurowski, D.Seese:
\newblock An algebraic theory of graph reduction.
\newblock Jour. of the association of computing machinery,  1134--1164, (1993)

\bibitem{goodman}
S.E.Goodman, S.T.Hedetniemi:
\newblock Sufficient conditions for a graph to be hamiltonian.
\newblock Journal of Combinatorial Theory,  175--180, (1974)

\bibitem{akiyama}
T.Akiyama, T.Nishizeki, N.Saito:
\newblock \textsc{NP}-completeness of the hamiltonian cycle problem for biparte
  graphs.
\newblock Journal of Information Processing,  73--76, (1980)

\bibitem{chvatal}
V.Chv$\acute{a}$tal:
\newblock On hamilton's ideals.
\newblock Journal of Combinatorial Theory 12(B),  163--168, (1972)

\bibitem{chvataltough}
V.Chv$\acute{a}$tal:
\newblock Tough graphs and hamiltonian circuits.
\newblock Discrete Mathemetics,  215--228, (1973)

\bibitem{chvatalerdos}
V.Chv$\acute{a}$tal, P.Erdos:
\newblock A note on hamiltonian circuits.
\newblock Discrete Mathematics,  111--113, (1972)

\bibitem{Gordon}
V.S.Gordon, Y.L.Orlovich, F.Werner:
\newblock Hamiltonian properties of triangular grid graphs.
\newblock Discrete Mathematics,  6166 -- 6188, (2008)

\bibitem{lindergen}
W.F.Lindgren:
\newblock An infinite class of hypohamiltonian graphs.
\newblock The American Mathematical Monthly,  1087--1089, (1967)

\bibitem{shih}
W.K.Shih, T.C.Chern, W.L.Hsu:
\newblock An \textsc{O}($n^{2}$log n) algorithm for the hamiltonian cycle
  problem on circular-arc graphs.
\newblock SIAM Journal on Computing,  1026--1046, (1992)

\end{thebibliography}
\end{document}